\documentclass[11pt, a4paper]{article}
\pdfoutput=1 

\usepackage{fullpage}

\usepackage{mdframed}

\usepackage[utf8]{inputenc}

\usepackage{libertine}

\usepackage{amsfonts}
\usepackage{amsmath}
\usepackage{amssymb}
\usepackage{amsthm}
\usepackage{float}

\usepackage{cmap}
\usepackage{hyperref}
\usepackage[svgnames]{xcolor}
\hypersetup{colorlinks={true},urlcolor={blue},linkcolor={DarkBlue},citecolor=[named]{DarkGreen}}
\usepackage[authoryear,square]{natbib}

\usepackage{xspace}
\usepackage{fancyhdr}
\usepackage{tcolorbox}

\usepackage{microtype}
\usepackage[capitalise,nameinlink]{cleveref}
\usepackage{enumitem}

\usepackage{doi}

\usepackage{authblk}

\usepackage{booktabs} 
\usepackage[ruled,linesnumbered]{algorithm2e} 
\usepackage{setspace}

\SetAlFnt{\small}
\SetAlCapFnt{\small}
\SetAlCapNameFnt{\small}
\SetAlCapHSkip{0pt}
\IncMargin{-\parindent}

\usepackage{tikz}  
\usetikzlibrary{arrows}
\usetikzlibrary{patterns,snakes}
\usetikzlibrary{decorations.shapes}
\tikzstyle{overbrace text style}=[font=\tiny, above, pos=.5, yshift=5pt]
\tikzstyle{overbrace style}=[decorate,decoration={brace,raise=5pt,amplitude=3pt}]
\usetikzlibrary{shapes.geometric}

\usepackage{authblk}
\usepackage{xspace}

\theoremstyle{definition}

\theoremstyle{plain}
\newtheorem{theorem}{Theorem}[section]
\newtheorem{lemma}[theorem]{Lemma}

\newtheorem{conjecture}[theorem]{Conjecture}

\theoremstyle{definition}

\newcommand{\xbf}{\ensuremath{\mathbf{x}}\xspace}
\newcommand{\ybf}{\ensuremath{\mathbf{y}}\xspace}
\newcommand{\pbf}{\ensuremath{\mathbf{p}}\xspace}

\newcommand{\tbf}{\ensuremath{\mathbf{t}}\xspace}

\newcommand{\argy}[1]{{\color{magenta}[Argy: #1]}}

\allowdisplaybreaks

\title{\bf Heterogeneous Facility Location with Limited Resources}

\author[1]{Argyrios Deligkas}
\author[2]{Aris Filos-Ratsikas}
\author[3]{Alexandros A. Voudouris}

\affil[1]{Royal Holloway University of London, United Kingdom}
\affil[2]{University of Liverpool, United Kingdom}
\affil[3]{University of Essex, United Kingdom}

\date{}

\begin{document}

\maketitle
\thispagestyle{empty}

\begin{abstract}
We initiate the study of the heterogeneous facility location problem with limited resources. We mainly focus on the fundamental case where a set of agents are positioned in the line segment $[0,1]$ and have approval preferences over two available facilities. A mechanism takes as input the positions and the preferences of the agents, and chooses to locate a single facility based on this information. We study mechanisms that aim to maximize the social welfare (the total utility the agents derive from facilities they approve), under the constraint of incentivizing the agents to truthfully report their positions and preferences. We consider three different settings depending on the level of agent-related information that is public or private. For each setting, we design deterministic and randomized strategyproof mechanisms that achieve a good approximation of the optimal social welfare, and complement these with nearly-tight impossibility results.
\end{abstract}

\setcounter{page}{1}

\section{Introduction} \label{sec:intro}
The truthful facility location problem is one of the most prominent paradigms in environments with strategic participants, and it was in fact the prototypical problem used by \cite{PT09} to put forward their very successful research agenda of \emph{approximate mechanism design without money} about a decade ago. Since then, the problem has been extensively studied in the literature of theoretical computer science and artificial intelligence, with a plethora of interesting variants emerging over the years. Among those, one particularly meaningful variant, which captures several important scenarios, is that of \emph{heterogeneous facility location}, introduced by \cite{FJ15} and studied notably by \cite{SV15,SV16}, \cite{anastasiadis2018heterogeneous}, \cite{fong2018facility}, \cite{chen2020facility} and \cite{li2020strategyproof}. In this setting, there are multiple facilities, and each of them plays a different role -- for example, a library and a basketball court. Consequently, the preferences of the agents for the possible outcomes do not only depend on the \emph{location} of the facility (as in the original model of \cite{PT09}), but also on the \emph{type} of the facility. As a result, the mechanism design problem now becomes far more challenging.\footnote{In particular, the preference domain is no longer \emph{single-peaked}, and therefore maximizing the happiness of the agents cannot be achieved by simple median mechanisms.}

While the literature on heterogeneous facility location is quite rich by this point, there is a fundamental setting that has surprisingly eluded previous investigations. In particular, all previous works have considered the case of multiple (predominantly two) facilities which \emph{all} have to be located, based on the positions and the preferences of the agents. However, in many real-world applications, resources are {\em limited}, and therefore a decision has to be made about which {\em subset} of the facilities should be build and where. For instance, the governing body might have sufficient funds to build only one of two options, either a library or a basketball court. The decision must be made based on the preferences of the agents over the two facilities, but also on their positions, in a way that incentivizes the agents to reveal all their private information truthfully; this is clearly a challenging mechanism design problem.

\subsection{Our setting}

We initiate the study of the heterogeneous facility location problem with \emph{limited resources}. We focus on the most fundamental case where there are two facilities, and only one of them must be located somewhere in the line segment $[0,1]$. In particular, there is a set of agents, each of whom is associated with a {\em position} in $[0,1]$ and an {\em approval preference} over the facilities. An agent may approve one of the two facilities or both, and obtains positive {\em utility}\footnote{We remark that in several facility location settings (e.g., see \citep{PT09,Lu09,Lu10}), the agents are associated with costs instead of utilities. In the literature of heterogeneous facility problems however, the setting is commonly defined in terms of utilities, as there is no meaningful way of assigning a cost to undesirable outcomes, such as a facility which the agent does not approve.} only if a facility that she approves is built; otherwise, she has zero utility irrespectively of her position. 

Our goal is to design {\em strategyproof mechanisms} that choose and locate a single facility, so as to maximize the {\em social welfare} (the total utility of the agents) and incentivize the agents to truthfully report their private information. We study the following three settings depending on the level of information about the positions and the preferences of the agents that is assumed to be public or private.
\begin{itemize}
\item {\em General} setting: Both the positions and the preferences are private information of the agents.

\item {\em Known-preferences} setting: The positions are private information of the agents, whereas the preferences are public information.

\item {\em Known-positions} setting: The preferences are private information of the agents, whereas the positions are public information.
\end{itemize}
We measure the performance of a strategyproof mechanism by its \emph{approximation ratio}, defined as the worst-case ratio over all instances of the problem between the maximum possible social welfare and the social welfare achieved by the mechanism. For each of the aforementioned settings, we derive upper and lower bounds on the achievable approximation ratio of strategyproof mechanisms. An overview of our results can be found in Table~\ref{table:results}.

\begin{table}[t]
 \begin{center}

 \begin{tabular}{ l || c | c }
 & Deterministic & Randomized \\
 \hline
 General & 2 & $(1, 2]$\\
 \hline
 Known-preferences & 2 & $[4/3^{\star},4/3]$\\
 \hline
 Known-positions & $[13/11, 2]$ & $(1, 3/2]$
 \end{tabular}
 \caption{Overview of our results for deterministic and randomized strategyproof mechanisms. The lower bound $4/3$ (marked with $\star$) in the known-preferences setting holds only for the class of Random-Median mechanisms defined in Section~\ref{sec:preferences}. For the general and known-preference settings, the bound of $2$ also holds for the more general case where we can choose $k$ out of $m\geq 2$ facilities, for appropriate values of $k$ and $m$.}
\label{table:results}
\end{center}
\end{table}

\subsection{Discussion of our results}
We start our investigation by studying deterministic mechanisms in the general setting, where we show that a simple group-strategyproof mechanism, which we call \emph{Middle mechanism}, achieves an approximation ratio of $2$ (Theorem~\ref{thm:Middle}); the same guarantee extends to the other two settings we consider. We complement this result by showing a lower bound of $2$ on the approximation ratio of \emph{any} deterministic strategyproof mechanism, even when the preferences of the agents are assumed to be known (Theorem~\ref{thm:known-deterministic-lower}). Combining these two results, we completely resolve the problem of identifying the best possible deterministic strategyproof mechanism for both the general and the known-preferences settings. For the known-positions setting, we show that there is no deterministic strategyproof mechanism with approximation ratio better than $13/11$ (Theorem~\ref{thm:13over11-deter}).  

We also consider randomized mechanisms, and provide improved approximation guarantees for both the known-preferences and the known-positions settings. More specifically, for the known-preferences setting we derive a novel universally group-strategyproof mechanism, termed \emph{Mirror mechanism}, which achieves an approximation ratio of $4/3$ (Theorem~\ref{thm:mirror}). This mechanism is in fact a member of a larger class of universally group-strategyproof mechanisms, and as we prove, it is the best possible mechanism in this class (Theorem~\ref{thm:random-median}). 
For the known-positions setting, we prove that the well-known {\em Random Dictatorship} mechanism, equipped with a carefully chosen tie-breaking rule for the agents that approve both facilities, is a universally group-strategyproof mechanism (Theorem~\ref{thm:RD-sp}) and achieves an approximation ratio of $3/2$ (Theorem~\ref{thm:RD}). 

Finally, we make initial progress in more general settings with $m\geq 2$ facilities, from which we can choose to locate $k < m$. We consider three variations based on whether the utility of each agent is determined by all the facilities she approves, or by the one that is the closest to or the farthest away from her position. For such utility classes, we show that an adaptation of the Middle mechanism still has an approximation ratio of $2$ in the general setting, it is group-strategyproof for $k=1$, but it is only strategyproof for $k \geq 2$ (Theorem~\ref{thm:Middle-km} and Lemma~\ref{lem:km-Middle-not-gsp-kgeq2}). We complement this result by showing that when $k \leq 2m$ it is impossible to do better, even when the preferences of the agents are known  (Theorem~\ref{thm:km-deterministic-LB-2}). 

\subsection{Related Work}

As we mentioned earlier, the literature on truthful facility location is long and extensive; here, we discuss only those works that are most closely related to our setting. The fundamental difference between our work and virtually all of the papers on heterogeneous facility location is that they consider settings with two facilities, where \emph{both} facilities have to be built, and the utility/cost of an agent is calculated with respect to the closest or the farthest among the two.  

In particular, \cite{chen2020facility} consider a setting in which agents have approval preferences over the facilities, similarly to what we do here, and for which the positions of the agents are known. \cite{li2020strategyproof} consider a more general metric setting along the lines of \cite{chen2020facility}, and design a deterministic mechanism which improves upon the result of \cite{chen2020facility} when the metric is a line. \cite{fong2018facility} consider a setting in which the agents have fractional preferences in $(0,1)$; similarly to us, besides studying the general setting, they also consider restricted settings with known preferences or known positions. \cite{SV15,SV16} consider a discrete setting, where the agents are positioned on the nodes of a graph, and the facilities must be located on different nodes. \citet{FJ15} were the first to study heterogeneous facility location, by presenting a ``hybrid'' model combining the standard facility location problem with the obnoxious facility location problem \citep{CWZ11,CWZ13}. This setting was extended by \cite{anastasiadis2018heterogeneous}, who allowed agents to be indifferent between whether a facility would be built or not. \cite{DLLX19} study a setting where the goal is to locate two facilities under the constraint that the distance between the locations of the facilities is at least larger than a predefined bound. 

\cite{LWZ20} study a conceptually similar but fundamentally different facility location problem under budget constraints. In their setting, the facilities are \emph{strategic} and need to be compensated monetarily in order for them to be built; the goal is to maximize an aggregate objective given that the total payment is below a predefined budget. \cite{LWZ20} study a conceptually similar but fundamentally different facility location problem under budget constraints. In their setting, the facilities are \emph{strategic} and need to be be compensated monetarily in order for them to be built; the goal is to maximize an aggregate objective given that the total payment is below a predefined budget. Besides these works, there is long literature of (homogeneous) facility location, studying 
different objectives~\citep{AFPT10,CFT16,FSY,FW}, 
multiple facilities~\citep{escoffier2011strategy,fotakis2013winner,Lu09,Lu10}, 
different domains~\citep{schummer2002strategy,tang2020characterization,sui2013analysis,sui2015approximately}, different cost functions \citep{FLZZ, fotakis2013strategyproof}, and several interesting variants~\citep{golomb2017truthful,kyropoulou2019mechanism,zhang2014strategyproof,filos2020approximate}.

\section{Preliminaries} \label{sec:prelims}
We consider a facility location setting with a set $N$ of $n$ {\em agents} and two {\em facilities}; we will discuss extensions to settings with more than two facilities in Section~\ref{sec:extensions}. Every agent $i \in N$ has a {\em position} $x_i \in [0,1]$; let $\xbf=(x_i)_{i \in N}$ be the {\em position profile} consisting of the positions of all agents. Furthermore, every agent $i \in N$ also has an {\em approval preference} (or, simply, {\em preference}) $\tbf_i = \{0,1\}^2$ over the two facilities, where $t_{ij}=1$ denotes that the agent {\em approves} facility $j \in \{1, 2\}$ and $t_{ij}=0$ denotes that she {\em does not approve} facility $j$; let $\tbf = (\tbf_i)_{i \in N}$ be the {\em preference profile} consisting of the preferences of all agents. Let $I = (\xbf, \tbf)$ denote an instance of this setting. 

Given an instance $I = (\xbf, \tbf)$, our goal is to {\em choose} and {\em locate} a single facility so as to optimize some objective function that depends on both the distances of the agents from the facility location and on whether they approve the chosen facility.
In particular, if facility $j \in \{1, 2\}$ is chosen to be located at $y \in [0,1]$, the {\em utility} of every agent $i \in N$ is defined to be $u_i(j, y | I) =  t_{ij} \cdot \big(1-d(x_i,y)\big)$, where $d(x_i,y)=|x_i-y|$ is the distance between $x_i$ and $y$. Then, the {\em social welfare} is the sum of the utilities of all agents: 
$$W(j, y | I) := \sum_{i \in N} u_i(j, y | I).$$
We denote the {\em optimal} social welfare for instance $I$ as $W^*(I) := \max_{(j,y)} W(j,y|I)$.

A {\em mechanism} $M$ takes as input an instance $I=(\xbf, \tbf)$ consisting of the position and preference profiles of the agents, and outputs an {\em outcome} $M(I) = (j_M, y_M)$ consisting of a facility $j_M \in \{1, 2\}$ that is to be placed at $y_M \in [0,1]$.  
The {\em approximation ratio} $\rho(M)$ of $M$ is defined as the worst-case ratio (over all possible instances) between the optimal social welfare and the social welfare of the outcome chosen by the mechanism, that is, 
$$\rho(M) = \sup_I \frac{W^*(I)}{W(M(I)|I)}.$$
A mechanism is {\em strategyproof} if it is in the best interest of every agent to report their true position and preferences, irrespectively of the reports of the other agents. Formally, a mechanism $M$ is strategyproof if, for every pair of instances $I=(\xbf, \tbf)$ and $I'=((x_i',\xbf_{-i}), (\tbf_i',\tbf_{-i}))$ in which only a single agent $i$ misreports a different position and preferences, it holds that 
$$u_i(M(I)|I) \geq u_i(M(I')|I).$$

Besides mechanisms that deterministically select a facility and its location, we will also study {\em randomized} mechanisms, which choose the outcome according to probability distributions. In particular, a randomized mechanism locates each facility $j \in \{1, 2\}$ at $y \in [0,1]$ with some probability $p_j(y)$ such that $\sum_{j \in \{1, 2\}} \int_0^1 p_j(y)dy = 1$. Denoting by $\pbf = (p_1, p_2)$ the probability distribution (for both facilities) used by the mechanism, the {\em expected utility} of every agent $i \in N$ is computed as 
$$u_i(\pbf|I) = \sum_{j \in \{1, 2\}} t_{ij} \cdot \int_0^1 (1-|x_i-y|)  \cdot p_j(y)dy.$$
A randomized mechanism is {\em strategyproof in expectation} if no agent can increase her {\em expected utility} by misreporting. Also, we say that a randomized mechanism is {\em universally strategyproof} if it is a probability distribution over deterministic strategyproof mechanisms. Clearly, a universally strategyproof mechanism is strategyproof in expectation, but the converse is {\em not} necessarily true.  

We will also discuss about mechanisms that are resilient to misreports by coalitions of agents. 
In particular, a mechanism is {\em group-strategyproof} if no coalition of agents can simultaneously misreport such that the utility of {\em every} agent in the coalition strictly increases. 

We are interested in mechanisms that satisfy strategyproofness properties (like the ones discussed above) and at the same time achieve an as low as possible approximation ratio (that is, an approximation ratio as close as possible to $1$). In our technical analysis in the upcoming sections, we will distinguish between the following settings:
\begin{itemize}
\item In the {\em general setting}, the agents can misreport both their positions and preferences.
\item In the {\em known-preferences setting}, the preferences of the agents are assumed to be known and the agents can misreport only their positions.
\item In the {\em known-positions setting}, the positions of the agents are assumed to be known and the agents can misreport only their preferences.
\end{itemize}
Observe that positive results (i.e., (group-)strategyproof mechanisms with proven approximation guarantees) for the general setting are also positive results for the known-preferences and known-positions settings. Moreover, negative results (i.e., lower bounds on the approximation of (group-)strategyproof mechanisms) for the restricted settings are also negative results for the general setting. Finally, results (positive or negative) for one of the two restricted settings do not imply anything for the other restricted setting. 


\section{General setting}\label{sec:general}
We start the presentation of our technical results by focusing on the general setting; recall that in this setting the agents can misreport both their positions and their preferences. Due to the structure of the problem, which combines voting (based on the preferences of the agents) and facility location (based on the positions of the agents), it is natural to wonder whether simple adaptations of the median mechanism (which is known to be strategyproof and optimal for the original single-facility location problem) lead to good solutions. For example, we could define mechanisms that locate the majority-winner facility (breaking ties in a consistent way) at the median among the agents that approve it, or at the overall median agent. Unfortunately, it is not hard to observe that the first mechanism is {\em not} strategyproof, while the second one has an approximation ratio that is linear in the number of agents. 

Luckily, there is an even simpler deterministic mechanism that is group-strategyproof and achieves an approximation of at most $2$ in the general setting. In the next section, we will further show that this mechanism is best possible among all deterministic strategyproof mechanisms in terms of approximation, even when the preferences of the agents are known.

\medskip

\begin{tcolorbox}[title=Middle mechanism (MM)]
\begin{enumerate}
\item 
Count the number $n_j$ of agents that approve each facility $j \in \{1,2\}$.
\item 
Locate the most preferred facility at location $\frac{1}{2}$, breaking ties arbitrarily.
\end{enumerate}
\end{tcolorbox}

\begin{theorem}
\label{thm:Middle}
The Middle mechanism is group-strategyproof and has an approximation ratio of at most $2$. 
\end{theorem}

\begin{proof}
Consider any instance $I=(\xbf, \tbf)$. 
To show that the mechanism is group-strategyproof, first observe that the positions of the agents are not taken into account when deciding which facility to locate and where. Hence, no agent has a reason to misreport her position. It remains to argue that there exists no group of agents who can all strictly increase their utility by misreporting their preferences. To this end, assume that facility $j\in \{1,2\}$ is chosen to be placed at $1/2$. Observe that the utility of any agent that approves $j$ is maximized subject to the constraint that the chosen facility is always placed at $1/2$. Hence, such agents would not have incentive to participate in a misreporting coalition. Moreover, the count $n_j$ of facility $j$ would only increase if any group of agents that truly disapprove facility $j$, misreport that they approve it. Hence, the outcome would not change in such a case, this proving that is indeed group-strategyproof.

We now focus on the approximation ratio of the mechanism. Let $j$ be the facility chosen by the mechanism, and let $o$ be the optimal facility. We make the following simple observations:
\begin{itemize}
\item 
Since the facility is placed at $1/2$, every agent $i$ that approves $j$ has utility at least $1/2$.

\item 
By the definition of the mechanism, we have that $n_j \geq n_o$. 

\item 
Since the maximum utility of any agent is $1$, we have that $W^*(I) \leq n_o$.
\end{itemize} 
Putting all of these together, we have:
\begin{align*}
W(\text{MM}(I)|I) \geq \frac{1}{2} n_j \geq \frac{1}{2} n_o \geq \frac{1}{2} W^*(I),
\end{align*}
and the bound on the approximation ratio follows.
\end{proof}

\section{Known-preferences setting} \label{sec:preferences}
Here, we focus on the known-preferences setting, where we assume that the agents can only strategize over their positions. Our first result is a lower bound of $2$ on the approximation ratio of any strategyproof deterministic mechanism, thus proving that the Middle mechanism presented in the previous section is best possible for the general and the known-preferences settings. 

\begin{figure}[t]
\centering
\includegraphics[scale=0.74]{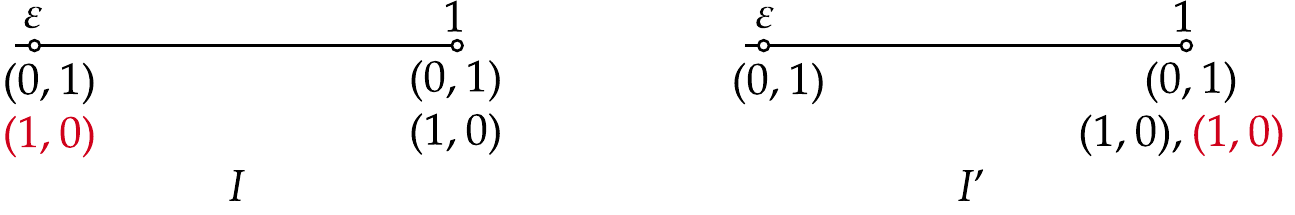}
\caption{The two instances used in the proof of Theorem~\ref{thm:known-deterministic-lower}, which differ only on the position of an agent with preferences $(1,0)$ marked in red.}
\label{fig:lb2}
\end{figure}

\begin{theorem}
\label{thm:known-deterministic-lower}
In the known-preferences setting, there is no deterministic strategyproof mechanism with approximation ratio better than $2-\delta$, for any $\delta > 0$.
\end{theorem}

\begin{proof}
Consider an arbitrary deterministic strategyproof mechanism and the following instance $I$ depicted in Figure~\ref{fig:lb2}. There are four agents, two with preferences $(0,1)$ and two with preferences $(1,0)$. One agent of each type is positioned at some $\varepsilon \in (0, 1/2)$ and the other is positioned at $1$. Without loss of generality, we can assume that the mechanism chooses to locate facility $2$; the welfare of the agents in this instance is maximized no matter where the facility is actually located.

Now consider a second instance $I'$ that is obtained from $I$ when only the agent $i$ with preference $(1,0)$ that is positioned at $\varepsilon$ is moved to $1$. Since the mechanism is strategyproof, it must choose to locate facility $2$ in instance $I'$ as well; otherwise, agent $i$ would prefer to misreport her position in instance $I$ as $1$, thus leading to instance $I'$ and the selection of facility $1$, which would increase her utility from $0$ to positive. However, the welfare from locating facility $2$ in instance $I'$ is at most $1+\varepsilon$ (no matter where it is located), whereas the optimal welfare is equal to $2$, achieved when facility $1$ is located at $1$. The bound on the approximation ratio follows by selecting $\varepsilon$ to be arbitrarily small.
\end{proof}

Next, we turn our attention to randomization and consider the class of {\em Random-Median} mechanisms. Every mechanism in this class operates by first randomly choosing one of the facilities based on the preferences the agents, which is then located at the median among the agents that approve it. So, different choices of the probability distribution according to which the facility is chosen lead to different Random-Median mechanisms. It is not hard to observe that all such mechanisms are universally group-strategyproof. 

\begin{lemma} \label{lem:Random-Median-SP}
Every Random-Median mechanism $M$ is universally group-strategyproof.
\end{lemma}

\begin{proof}
The lemma follows directly by the following two facts: 
(1) The choice of the facility to be located is made only based on the preferences of the agents, which are assumed to be known, and thus cannot be manipulated.
(2) Given the facility, the location is chosen to be the position of the median agent among the ones that approve it, which is known to be a strongly strategyproof mechanism. 
\end{proof}

Probably the simplest Random-Median mechanism one can think of is the Proportional mechanism defined below, which selects every facility with probability proportional to the number of agents that approve it. By exploiting the definition of this probability, we can show that the Proportional mechanism has an approximation of $(1 + \sqrt{3})/2 \approx 1.366$, thus significantly improving upon the bound of $2$ achieved by deterministic mechanisms.

\medskip

\begin{tcolorbox}[title=Proportional mechanism]
\begin{enumerate}
\item
Count the number $n_j$ of agents that approve each facility $j \in \{1,2\}$.
\item 
With probability $\frac{n_j}{n_1+n_2}$ locate facility $j$ at the median among the agents that approve it.
\end{enumerate}
\end{tcolorbox}

\begin{theorem}
\label{thm:Proportional}
In the known-preferences setting, the Proportional mechanism is universally group-strategyproof and has an approximation ratio $(1 + \sqrt{3})/2 \approx 1.366$. 
\end{theorem}

\begin{proof}
Since the mechanism is Random-Median, it is universally group-strategyproof due to Lemma~\ref{lem:Random-Median-SP}. To bound the approximation ratio, let $W_j$ be the welfare of the agents that approve facility $j$ when it is chosen (and located at the median of those agents). Without loss of generality, assume that $W_1 \geq W_2$. We also have that $W_1 \leq n_1$ since the maximum possible utility of any agent is $1$. Furthermore, we have that $W_2 \geq n_2/2$. To see why this is the case, consider the agents that approve facility 2 in pairs, where one is on the left of the median (among those that approve facility $2$) and the other is on the right of the median, and observe that the total utility of this pair of agents is at least $1$. Since there are $n_2/2$ such pairs, the claim follows. 
The approximation ratio is
\begin{align*}
\rho(\text{Proportional}) 
&= \frac{W_1}{\frac{n_1}{n_1+n_2}W_1 + \frac{n_2}{n_1+n_2}W_2} 
= \frac{1}{\frac{n_1}{n_1+n_2} + \frac{n_2}{n_1+n_2}\cdot \frac{W_2}{W_1}} \\
&\leq \frac{1}{ \frac{n_1}{n_1+n_2} + \frac{n_2}{n_1+n_2} \cdot \frac{n_2/2}{n_1}}
= \frac{2n_1^2 + 2n_1 n_2}{2n_1^2 + n_2^2}.
\end{align*}
Now, let $y=n_1/n_2$ and observe that, since $n_1 \geq W_1 \geq W_2 \geq n_2/2$, it must be that $y \geq 1/2$. By dividing the last expression above by $n_2^2$, we obtain:
\begin{align*}
\rho(\text{Proportional}) \leq \frac{2(n_1/n_2)^2 + 2(n_1/n_2)}{2(n_1/n_2)^2 + 1} = \frac{2y^2+2y}{2y^2+1}.
\end{align*}
Hence, in order to bound the approximation ratio of the mechanism it suffices to maximize the function $\frac{2y^2+2y}{2y^2+1}$ subject to the constraint $y \geq 1/2$. It is not hard to observe that the maximum value is $(1+\sqrt{3})/2 \approx 1.366$ for $y^* = 1/(\sqrt{3}-1)$, thus proving the upper bound on the approximation ratio of the Proportional mechanism.
\end{proof}

We can further improve upon the bound of the Proportional mechanism, by defining the slightly more involved Mirror mechanism defined below, which uses a probability distribution that is a piecewise function of the numbers of agents that approve the two facilities. Following along the lines of the proof of Theorem~\ref{thm:Proportional}, we can show that the Mirror mechanism has an approximation ratio of $4/3$. 

\medskip

\begin{tcolorbox}[title=Mirror mechanism]
\begin{enumerate}
\item
Count the number $n_j$ of agents that approve each facility $j \in \{1,2\}$; wlog assume $n_1 \geq n_2$ (otherwise switch $n_1$ with $n_2$ below).
\item 
Let $\alpha := \frac{3n_1-2n_2}{4n_1-2n_2}$.
\item 
Choose facility $1$ with probability $\alpha$, and facility $2$ with probability $1-\alpha$.
\item Locate the chosen facility at the median among the agents that approve it.
\end{enumerate}
\end{tcolorbox}

\begin{theorem}
\label{thm:mirror}
In the known-preferences setting, the Mirror mechanism is universally group-strategyproof and has an approximation ratio of $4/3$. 
\end{theorem}

\begin{proof}
Since the mechanism is Random-Median, it is universally group-strategyproof due to Lemma~\ref{lem:Random-Median-SP}.
To bound the approximation ratio, let $W_j$ be the welfare of the agents that approve facility $j$ when it is chosen (and located at the median of those agents). Observe that for any facility $j \in \{1,2\}$ it holds that $W_j \leq n_j$ since the maximum possible utility of any agent is $1$, and $W_j \geq n_j/2$ following the same reasoning as in the proof of Theorem~\ref{thm:Proportional}. We distinguish between the following two cases:
\begin{itemize}
\item
$W_1 \geq W_2$. Then, the approximation ratio is:
\begin{align*}
\rho(\text{Mirror}) 
&= \frac{W_1}{\alpha \cdot W_1 + (1-\alpha) \cdot W_2}
= \frac{1}{ \frac{3n_1-2n_2}{4n_1-2n_2} + \frac{n_1}{4n_1-2n_2} \cdot \frac{W_2}{W_1}} \\
&\leq \frac{1}{ \frac{3n_1-2n_2}{4n_1-2n_2} + \frac{n_1}{4n_1-2n_2} \cdot \frac{n_2/2}{n_1}} 
= \frac43.
\end{align*}

\item 
$W_2 > W_1$. We have:
\begin{align*}
\rho(\text{Mirror})
&= \frac{W_2}{\alpha \cdot W_1 + (1-\alpha) \cdot W_2}
= \frac{1}{ \frac{3n_1-2n_2}{4n_1-2n_2} \cdot \frac{W_1}{W_2} + \frac{n_1}{4n_1-2n_2}} \\
&\leq \frac{1}{ \frac{3n_1-2n_2}{4n_1-2n_2} \cdot \frac{n_1/2}{n_2} + \frac{n_1}{4n_1-2n_2}} 
= \frac{2n_2(4n_1-2n_2)}{ 3n_1^2 }. 
\end{align*}
It is not hard to observe that this last expression is maximized to $4/3$ when $n_1=n_2$.
\end{itemize}
Hence, in any case the approximation ratio of the mechanism is at most $4/3$.
\end{proof}

We conclude this section by showing that the Mirror mechanism is best possible among all Random-Median mechanisms in terms of approximation.

\begin{theorem}
\label{thm:random-median}
In the known-preferences setting, the approximation ratio of any Random-Median mechanism is at least $4/3-\delta$, for any $\delta > 0$.
\end{theorem}

\begin{proof}
Consider an arbitrary Random-Median mechanism and the following instance $I$ (which was also used in the proof of Theorem~\ref{thm:known-deterministic-lower}. There are four agents, two with preferences $(0,1)$ and two with preferences $(1,0)$. One agent of each type is positioned at $\varepsilon \in (0,1/2)$ while the other is positioned at $1-\varepsilon$. Due to symmetry, we can assume that the agents located at $\varepsilon$ are the medians for the two facilities. The mechanism randomly chooses one of the facilities. Without loss of generality, we can assume that it chooses facility $1$ with some probability $p \leq 1/2$ and facility $2$ with probability $1-p$. Hence, facility $1$ is located at $\varepsilon$ with probability $p$. 

Now consider the instance $I'$ that is obtained from $I$ by moving the agent $i$ at $1-\varepsilon$ with preference $(1,0)$ to $\varepsilon$. We claim that the mechanism must choose facility $1$ with probability $p' \leq p \leq 1/2$ in this new instance. Suppose that this is not the case and the mechanism chooses facility $1$ with probability $p' > p$ in $I'$. Since the expected utility of agent $i$ is $p\cdot 2\varepsilon$ in $I$, she would have incentive to misreport her position as $\varepsilon$ so that facility $1$ is chosen with probability $p'$ and her expected utility is increased to $p' \cdot 2\varepsilon$. This contradicts the fact that the mechanism is strategyproof in expectation (since it is universally strong group-strategyproof). 

In instance $I'$, the maximum welfare we can achieve by placing facility $1$ is $2$ (when it is placed at $\varepsilon$), and by placing facility $2$ is $1+\varepsilon$ (no matter where it is located). Hence, since the maximum possible expected social welfare achieved by the mechanism is 
$2p' + (1+\varepsilon)(1-p') \leq \frac{3+\varepsilon}{2}$, the approximation ratio of the mechanism is at least $\frac{4}{3+\varepsilon} \geq 4/3-\delta$, for any $\delta > \frac{4\varepsilon}{9+3\varepsilon}$.
\end{proof}


\section{Known-positions setting}\label{sec:positions}
We now turn our attention to the known-positions setting, in which the positions of the agents are fixed, and thus the agents can misreport only their preferences. Our first result is a lower bound of $13/11$ on the approximation ratio of any deterministic strategyproof mechanism.

\begin{figure}[t]
\centering
\includegraphics[scale=0.74]{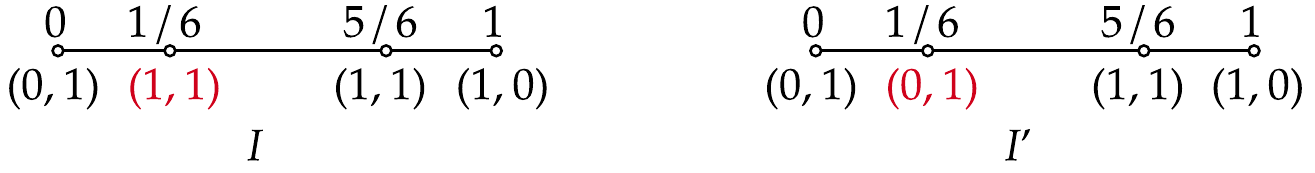}
\caption{The two instances used in the proof of Theorem~\ref{thm:13over11-deter}, which differ only on the preference of the agent positioned at $1/6$ marked in red.}
\label{fig:lb13-11}
\end{figure}

\begin{theorem}
\label{thm:13over11-deter}
In the known-positions setting, there is no deterministic strategyproof mechanism with approximation ratio smaller than $13/11$.
\end{theorem}

\begin{proof}
Suppose towards a contradiction that there exists a deterministic strategyproof mechanism that has an approximation ratio strictly smaller than $13/11$, and consider the following instance $I$ depicted in Figure~\ref{fig:lb13-11}. There is an agent with preferences $(0,1)$ positioned at $0$, an agent with preferences $(1,1)$ positioned at $1/6$, an agent with preferences $(1,1)$ positioned at $5/6$, and an agent with preferences $(1,0)$ positioned at $1$. Due to the symmetry of the instance, without loss of generality, we can assume that the mechanism chooses to place facility $1$ at some location $y \in [0,1]$.

If $y \leq 1/2$, the social welfare achieved by the mechanism is 
$$\bigg(1 - \left|y-\frac16\right|\bigg) + \bigg(1 - \left(\frac56-y\right) \bigg) + \bigg(1-(1-y) \bigg)
= \frac76 +2y - \left|y-\frac16\right| \leq \frac{11}{6}.$$
Since the optimal social welfare is $13/6$ (achieved by placing either facility $1$ at $5/6$ or facility $2$ at $1/6$), the approximation ratio of the mechanism is then $13/11$, contradicting the assumption that it is strictly smaller than $13/11$. Consequently, it must be $y > 1/2$.

Now consider the instance $I'$ that is obtained from $I$ by changing the preference of the agent at $1/6$ to $(0,1)$. In $I'$, the maximum possible social welfare one can hope to achieve by placing facility $1$ is $11/6$ (when it is located anywhere in the interval $[5,6,1]$ is $11/6$) and by placing facility $2$ is $13/6$ (when it is located at $1/6$). Hence, to have an approximation ratio strictly smaller than $13/11$, the mechanism must choose to locate facility $2$ in $I'$ at some position $z \in [0,1]$. Similarly to instance $I$, we can show that it must be $z < 1/2$ as otherwise the approximation ratio of the mechanism would be at least $13/11$. Hence, the agent positioned at $1/6$ with preferences $(1,1)$ in instance $I$ has incentive to misreport her preferences as $(0,1)$ so that the facility is located closer to her, and thus increase her utility. However, this contradicts the assumption that the mechanism is strategyproof, and the theorem thus follows. 
\end{proof}


Unfortunately, we have been unable to design any deterministic mechanism with approximation ratio strictly smaller than the bound of $2$ achieved by the Middle mechanism presented in Section~\ref{sec:general}, and leave it as a challenging open question. Instead, we continue by considering randomized mechanisms. Observe that the instances from Figure~\ref{fig:lb13-11} show that there is no randomised strategyproof mechanism with approximation ratio 1 for the known-positions setting. This is because no matter how the mechanism locates the facilities on instance $I$, it has to locate facility 2 on Instance $I'$, thus the agent positioned at $1/6$ has incentives to misreport her preferences. Next, we analyze the Random Dictatorship (RD) mechanism defined below. We will show that, in the known-positions setting, this mechanism is universally group-strategyproof and has an approximation ratio of $3/2$.

\medskip

\begin{tcolorbox}[title=Random Dictatorship (RD) mechanism]
\begin{enumerate}
\item
Choose an agent $i$ uniformly at random.
\item 
If agent $i$ approves a single facility $j$, then locate $j$ at $x_i$.
\item 
If agent $i$ approves both facilities, then locate the optimal one at $x_i$.
\end{enumerate}
\end{tcolorbox}

\begin{theorem}\label{thm:RD-sp}
In the known-positions setting, RD is universally group-strategyproof.
\end{theorem}

\begin{proof}
Clearly, RD is a uniform probability distribution over the set of all possible deterministic mechanisms each of which treats a different agent as dictator, and locates one of the facilities that the dictator approves at her position. The lemma follows since the agents cannot misreport their positions and no agent has any incentive to lie about which facilities she approves when she is chosen as a dictator; also note that any coalition of agents that does not include the dictator cannot change the outcome of the mechanism. 
\end{proof}

To show that the approximation ratio of RD is $3/2$, we will exploit a series of structural properties of the worst-case instances, in which the approximation ratio of the mechanism is maximized. In particular, we will show that every worst-case instance is characterized by the following three properties:
\begin{itemize}
\item There are no agents that approve both facilities (Lemma~\ref{lem:RD-11}).
\item Every agent that approves the non-optimal facility is positioned at $0$ or $1$ (Lemma~\ref{lem:RD-01}).
\item Every agent that approves the optimal facility is positioned at $0$ or some median position $x \in [0,1]$ or $1$ (Lemma~\ref{lem:RD-10}). 
\end{itemize}
To show these properties, we will start with an arbitrary instance and gradually change the preferences and the positions of the agents in a specific order so that the aforementioned properties are satisfied. Every change we make leads to a transformed instance in which the approximation ratio of RD does not decrease. It then suffices to define a worst-case instance satisfying these properties and bound the approximation ratio of the mechanism for this instance. This will be a function of a handful of variables representing the number of agents that are positioned at $\{0,x,1\}$ and approve one of the two facilities. 

The first property is quite easy to observe, given the definition of the mechanism. 

\begin{lemma}\label{lem:RD-11}
In a worst-case instance, all agents approve one facility.
\end{lemma}

\begin{proof}
Consider an arbitrary instance in which there is a set of agents that approve both facilities. If the mechanism chooses any such agent, then the optimal facility is placed at the position of this agent. By transforming the preference of the agent so that she only approves the optimal facility, the optimal welfare remains unaffected, whereas the welfare achieved by the mechanism can only decrease. This is because the agent does not contribute to the welfare gained when the dictator is an agent who approves only the non optimal facility. 
\end{proof}

Next, we show the second property. 

\begin{lemma}\label{lem:RD-01}
In a worst-case instance, every agent that approves a non-optimal facility is positioned at $0$ or $1$. 
\end{lemma}

\begin{proof}
Consider an arbitrary instance in which each agent approves a single facility (from Lemma~\ref{lem:RD-11}). Without loss of generality, assume that the optimal facility is $1$, and let $S$ be the set of agents that approve the non-optimal facility $2$. We order the agents in $S$ in terms of their positions, such that $x_1 \leq x_2 \leq ... \leq x_{|S|}$. We partition $S$ into two sets $L$ and $R$, such that $L$ consists of the first $\lceil |S|/2 \rceil$ agents in $S$ and $R = S \setminus L$. Observe that if the number of agents in $S$ is odd, the unique median agent in $S$ is included in $L$, whereas if the number of agents in $S$ is even, one of the two medians in $S$ is included in $L$, while the other is included in $R$. To prove the lemma, we claim that moving the agents in $L$ to $0$ from left to right, and the agents in $R$ to $1$ from right to left, leads to a sequence of instances such that the approximation ratio of RD does not decrease between consecutive instances. Due to symmetry, it suffices to prove this claim only for instances obtained by moving the agents in $L$. 

Let $m_L$ denote the median agent that is included in $L$ and observe that facility $1$ will remain optimal after moving any agent $\ell \in L \setminus \{m_L\}$. In particular, since the agents in $L \setminus \{\ell\}$ are not moved, their contribution to the maximum welfare achieved from facility $2$, i.e. when it is placed at the position of $m_L$, remains the same. On the other hand, the contribution of agent $\ell$ (who is moved) cannot increase as her distance from the median agent(s) in $S$ either remains the same (if $x_\ell=0$) or increases (if $x_\ell > 0$). Once all the agents in $L \setminus \{m_L\}$ have been moved to $0$, we can also show that moving agent $m_L$ to $0$ cannot increase the maximum possible welfare from facility $2$. This follows directly by the discussion below, where we show that the expected welfare of RD does not increase each time we move an agent in $L$, and thus the approximation ratio does not decrease. 

Suppose it is time to move agent $\ell \leq |L|$ to $0$, that is, it holds that $x_i=0$ for every $i < \ell$. The expected welfare of RD can be partitioned into the contribution $W_S$ of agents in $S$ (that approve only facility $2$) and the contribution $W_{\overline{S}}$ of the remaining agents (that approve only facility $1$). Clearly, changing the position of any agent in $S$ does not affect $W_{\overline{S}}$. We can write $W_S$ as 
\begin{align*}
W_S 
&= \frac{1}{n} \sum_{i \in S} \sum_{j \in S} \big( 1 - d(x_i,x_j) \big) \\
&=  \frac{1}{n} \bigg( 
\sum_{i \in S \setminus\{ \ell \}  } \sum_{j \in S \setminus\{ \ell \} } \big( 1 - d(x_i,x_j) \big) 
+ 2\sum_{i \in S \setminus\{ \ell \} } \big( 1 - d(x_i,x_\ell) \big)
+ 1
\bigg).
\end{align*}
After moving agent $\ell$ to $x_\ell'=0$, the first double sum in the above expression and the constant $1$ will remain unaffected. So, we will focus on the second sum, and show that it cannot increase. We define the set $S_{<\ell}$ of agents in $S$ before $\ell$ and the set $S_{> \ell}$ of agents in $S$ after $\ell$. By the definition of the set $L$, it holds that $|S_{< \ell}| \leq |S_{> \ell}|$, with the equality holding only in the case where the number of agents in $S$ is odd and $\ell = m_L$. Moreover, by the definition of agent $\ell$, it holds that $x_i=0$ for every $i \in S_{< \ell}$.
We now have that
\begin{align*}
\sum_{i \in S \setminus \{ \ell \} } \big( 1 - d(x_i,x_\ell) \big) 
&= \sum_{i \in S_{<\ell}}\big( 1 - d(x_i,x_\ell) \big)  + \sum_{i \in S_{> \ell}}\big( 1 - d(x_i,x_\ell) \big) \\
&= \sum_{i \in S_{<\ell}}\big( 1 - (x_\ell-x_i) \big) + \sum_{i \in S_{> \ell}}\big( 1 - (x_i-x_\ell) \big) \\
&=  \sum_{i \in S_{<\ell}}\big( 1 - (0-x_i)  \big) + \sum_{i \in S_{> \ell}}\big( 1 - (x_i-0)  \big)
+ |S_{>\ell}| - |S_{<\ell}| \\
&\geq  \sum_{i \in S_{<\ell}}\big( 1 - (x_\ell'-x_i) \big) + \sum_{i \in S_{> \ell}}\big( 1 - (x_i-x_\ell') \big) \\
&=  \sum_{i \in S \setminus \{ \ell \} } \big( 1 - d(x_i,x_\ell') \big), 
\end{align*}
as desired. This last sequence of equalities and inequalities also shows that moving $m_L$ to $0$ cannot increase the maximum possible social welfare from facility $2$, as claimed above.
\end{proof}

The third and last property follows by a proof similar to the one of the lemma above. 

\begin{lemma}\label{lem:RD-10}
Let $x \in [0,1]$ be the position of the median agent among those that approve the optimal facility in a worst-case instance. Then, every agent that approves the optimal facility is positioned at $0$ or $x$ or $1$.
\end{lemma}

\begin{proof}
Consider an arbitrary instance in which each agent approves a single facility (from Lemma~\ref{lem:RD-11}), and all agents approving the non-optimal facility are positioned at $0$ or $1$ (from Lemma~\ref{lem:RD-01}). Without loss of generality, assume that the optimal facility is $1$, and let $S$ be the set of agents that approve it. We will gradually transform this instance into one that satisfies the conditions of the lemma by appropriately moving the agents, such that each time we move an agent the approximation ratio of RD does not decrease. Let $L$ and $R$ be the sets of the non-median agents in $S$ who are at the left and at the right of the median agent(s) in $S$, respectively. Due to symmetry, it suffices to focus only on $L$ and show that the agents therein can be moved either to $0$ or $x$. We will also assume that facility $1$ remains optimal throughout the whole process;  observe that this is without loss of generality since if facility $2$ becomes the optimal one at some point of the process, then by replicating the procedure used in the proof of Lemma~\ref{lem:RD-01}, we can move all agents in $S$ to $0$ or $1$, thus obtaining the desired structure. 

Suppose the agents in $S$ are ordered in terms of their positions, such that $x_1 \leq x_2 \leq ... \leq x_{|S|}$. Let $\ell$ be the left-most agent in $L$ who is not positioned at $0$ or $x$, that is, $x_\ell < x$ and $x_i = 0$ for every $i < \ell$. We can write the optimal social welfare as 
\begin{align*}
W^* 
= \sum_{i \in S} \big( 1-d(x_i,x) \big) 
= \sum_{i \in S\setminus\{\ell\}} \big( 1-d(x_i,x) \big) + 1-(x - x_\ell)
= A + x_\ell.
\end{align*}
Let $W_S$ be the contribution of the agents in $S$ to the expected welfare achieved by RD, and let $W_{\overline{S}}$ denote the contribution of the remaining agents who approve facility $2$. We have 
\begin{align*}
W_S 
&= \frac{1}{n} \sum_{i \in S} \sum_{j \in S} \big( 1 - d(x_i,x_j) \big) \\
&=  \frac{1}{n} \bigg( 
\sum_{i \in S \setminus \{ \ell \} } \sum_{j \in S \setminus \{ \ell \} } \big( 1 - d(x_i,x_j) \big) 
+ 2\sum_{i \in S \setminus \{ \ell \} } \big( 1 - d(x_i,x_\ell) \big)
+ 1
\bigg).
\end{align*}
Let $S_{\leq \ell}$ be the set of agents in $S$ different than $\ell$ with position at most $x_\ell$, and let $S_{> \ell}$ be the set of agents in $S$ with position strictly larger than $\ell$. By the definition of $\ell$ (who belongs to $L$ and is the left-most agent that is not positioned at $0$ or $x$), it holds that  $|S_{\leq \ell}| < |S_{> \ell}|$. 
Therefore,  
\begin{align*}
W_S &= 
\frac{1}{n} \bigg( 
\sum_{i \in S \setminus \{ \ell \} } \sum_{j \in S \setminus \{ \ell \} }  \big( 1 - d(x_i,x_j) \big) 
+ 2\sum_{i \in S_{\leq \ell}}\big( 1 - d(x_i,x_\ell) \big)  
+ 2\sum_{i \in S_{> \ell}}\big( 1 - d(x_i,x_\ell) \big)
+ 1
\bigg) \\
&= 
\frac{1}{n} \bigg( 
\sum_{i \in S \setminus \{ \ell \} } \sum_{j \in S \setminus \{ \ell \} l} \big( 1 - d(x_i,x_j) \big) 
+ 2\sum_{i \in S_{\leq \ell}}\big( 1 + x_i \big)  
+ 2\sum_{i \in S_{> \ell}}\big( 1 - x_i \big)
+ 1
+ 2 x_\ell \big( |S_{> \ell}| - |S_{\leq \ell}| \big)
\bigg) \\
&= 
B + \frac{2}{n} x_\ell \big( |S_{> \ell}| - |S_{\leq \ell}| \big).
\end{align*}
Hence, the approximation ratio of RD is
\begin{align*}
\rho(RD) = \frac{A + x_\ell}{ W_{\overline{S}} + B + \frac{2}{n} x_\ell \big( |S_{> \ell}| - |S_{\leq \ell}| \big) }.
\end{align*}
Now observe that, since $|S_{> \ell}| - |S_{\leq \ell}| > 0$, the ratio is a monotonic function of $x_\ell$. If it is decreasing, we can move agent $\ell$ to $0$, and if it is increasing, we can move agent $\ell$ to the position $y \in (x_\ell, x]$ of the first agent that lies strictly to the right of $\ell$. Therefore, we can always move the left-most agent in $L$ who is not positioned at $0$ or $x$, to either $0$ or strictly to the right. At some point, this procedure will lead all agents in $L$ to be positioned either at $0$ or $x$, and symmetrically, all agents in $R$ to be positioned either at $x$ or $1$. 
\end{proof}

We are now ready to prove the bound on the approximation ratio of RD.

\begin{theorem}\label{thm:RD}
The approximation ratio of RD is $3/2$.
\end{theorem}

\begin{proof}
Consider a worst-case instance $I$ in which the optimal facility is $1$. From Lemmas~\ref{lem:RD-11}, \ref{lem:RD-01} and \ref{lem:RD-10}, we can assume that there are $n = \alpha_0 + \alpha_x + \alpha_1 + \beta_0 + \beta_1$ agents in total, such that:
\begin{itemize}
\item $\alpha_0$ agents approve only facility $1$ and are positioned at $0$;
\item $\alpha_x$ agents approve only facility $1$ and are positioned at some $x \in [0,1]$;
\item $\alpha_1$ agents approve only facility $1$ and are positioned at $1$;
\item $\beta_0$ agents approve only facility $2$ and are positioned at $0$;
\item $\beta_1$ agents approve only facility $2$ and are positioned at $1$.
\end{itemize} 
Since $x$ is the position of the median agent among those that approve facility $1$, we have that $\alpha_0 + \alpha_x \geq \alpha_1$ and $\alpha_x + \alpha_1 \geq \alpha_0$. We also have that $\beta_0 \geq \beta_1$. The optimal social welfare is achieved by placing facility $1$ at position $x$, and is equal to
\begin{align*}
W^*(I) =  \alpha_0 \cdot (1-x) +\alpha_x + \alpha_1 \cdot  x = \alpha_0 + \alpha_x + (\alpha_1-\alpha_0) x.
\end{align*}
Since facility $2$ is not optimal, we have that $W^* \geq \beta_0 \geq \beta_1$. 
The expected welfare of RD is
\begin{align*}
W(\text{RD}(I)) &= 
\frac{1}{n} \bigg(
\alpha_0 \big( \alpha_0 + \alpha_x(1-x) \big)
+ \alpha_x  \big( \alpha_0 (1-x) +\alpha_x + \alpha_1  x \big)
+ \alpha_1 \big( \alpha_x  x + \alpha_1 \big)
+ \beta_0^2 + \beta_1^2
\bigg) \\
&= \frac{1}{n} \bigg(
(\alpha_0 + \alpha_x)^2 + \alpha_1^2 
+2\alpha_x(\alpha_1-\alpha_0) x
+ \beta_0^2 + \beta_1^2
\bigg).
\end{align*}
Hence, by replacing $n = \alpha_0 + \alpha_x + \alpha_1 + \beta_0 + \beta_1$, the approximation ratio or RD can be written as the following function of $x$:
\begin{align*}
\rho(\text{RD}) 
&= (\alpha_0 + \alpha_x + \alpha_1 + \beta_0 + \beta_1) \cdot \frac{\alpha_0 + \alpha_x+ (\alpha_1-\alpha_0) x}
{(\alpha_0 + \alpha_x)^2 + \alpha_1^2  + \beta_0^2 + \beta_1^2  + 2\alpha_x(\alpha_1-\alpha_0) x}.
\end{align*}
It is now not hard to see that, since the factors of $x$ in the enumerator and the denominator are either both positive or negative, the ratio is a monotonic function in terms of $x$, and it thus attains its maximum value for $x=0$ or $x=1$. 

Therefore, we can further simplify the worst-case instance $I$ and assume that there are $n=\alpha_0 + \alpha_1 + \beta_0 + \beta_1$ agents in total, such that:
\begin{itemize}
\item $\alpha_0$ agents approve only facility $1$ and are positioned at $0$;
\item $\alpha_1$ agents approve only facility $1$ and are positioned at $1$;
\item $\beta_0$ agents approve only facility $2$ and are positioned at $0$;
\item $\beta_1$ agents approve only facility $2$ and are positioned at $1$;
\end{itemize} 
Without loss of generality, we assume that $\alpha_0 \geq \max \{ \alpha_1, \beta_0, \beta_1\}$, hence the optimal welfare is $W^*=\alpha_0$, achieved by placing facility $1$ at $0$. The expected social welfare of RD is 
$$W(\text{RD}(I)) = \frac{\alpha_0^2 + \alpha_1^2 + \beta_0^2 + \beta_1^2}{\alpha_0 + \alpha_1 + \beta_0 + \beta_1},$$
and thus the approximation ratio is
\begin{align*}
\rho(\text{RD}) = \frac{\alpha_0 (\alpha_0 + \alpha_1 + \beta_0 + \beta_1)}{\alpha_0^2 + \alpha_1^2 + \beta_0^2 + \beta_1^2}.
\end{align*}
By nullifying the partial derivatives of this function in terms of $\alpha_0$, $\alpha_1$, $\beta_0$ and $\beta_1$ we obtain a system of four equations, whose solution shows that the ratio is maximized to $3/2$ when $\alpha_0 = 3\alpha_1$ and $\alpha_1 = \beta_0 = \beta_1$.
\end{proof}

\section{Extensions to choosing $k$ out of $m$ facilities}
\label{sec:extensions}

So far we have exclusively focused on the fundamental case where there are two facilities and one of them must be located. In this section, we define and make initial progress for natural generalizations when there are $m \geq 2$ different facilities from which we can choose to locate $k < m$. There is a plethora of ways to define the utility of an agent. For instance, we can define it as the utility the agent derives from a subset of the facilities that are located and are among the ones she approves. This subset may include all such facilities, or just the facility that is the {\em closest} or the {\em farthest} from the agent's position. For such cases, it is quite easy to see that a straightforward adaptation of the Middle mechanism satisfies strategyproofness constraints and has an approximation ratio of at most $2$. In addition, by extending the proof of Theorem~\ref{thm:known-deterministic-lower}, we can show for particular values of $m$ and $k$ that this is the best-possible approximation among deterministic mechanisms, even when the preferences of the agents are assumed to be known.  

Formally, let $I = (\xbf, \tbf, m, k)$ be an instance with position profile $\xbf$, preference profile $\tbf$, and $m\geq 2$ facilities out of which we must chose and locate $k \geq 1$. Let $S$ a subset of $k$ facilities that are chosen to be located, and denote by $y_j \in [0,1]$ the location of facility $j \in S$; let $\ybf = (y_j)_{j \in S}$. Then, we can define the following three classes of utilities functions, to which we refer as {\em sum}, {\em min} and {\em max}: 
\begin{align*}
u_i^{\text{\em sum}}(S,\ybf|I) &= \sum_{j \in S} t_{ij} \cdot \big( 1 - d(x_i,y_j) \big) \\
u_i^{\text{\em min}}(S,\ybf|I) &= \max_{j \in S} \bigg\{ t_{ij} \cdot \big( 1 - d(x_i,y_j) \big) \bigg\} \\
u_i^{\text{\em max}}(S,\ybf|I) &= \min_{j \in S} \bigg\{ t_{ij} \cdot \big( 1 - d(x_i,y_j) \big) \bigg\}.
\end{align*}
For these three classes of utility functions, we will show that the $(k,m)$-Middle mechanism defined below is either group-strategyproof or strategyproof (depending on the number of facilities it must choose) and has an approximation ratio of $2$.

\medskip

\begin{tcolorbox}[title= {$(k,m)$}-Middle mechanism]
\begin{enumerate}
\item 
Count the number $n_j$ of agents that approve each facility $j \in \{1,\ldots,m\}$.
\item 
Locate the $k$ most preferred facilities at location $\frac{1}{2}$, breaking ties arbitrarily.
\end{enumerate}
\end{tcolorbox}

\begin{theorem} \label{thm:Middle-km}
For any utility class $C \in \{\text{sum}, \text{min}, \text{max}\}$, the $(k,m)$-Middle mechanism is group-strategyproof when $k=1$, strategyproof when $k \geq 2$, and has an approximation ratio of at most $2$. 
\end{theorem}

\begin{proof}
For $k=1$, the proof that the mechanism is group-strategyproof and has approximation ratio of at most $2$ follows directly by Theorem~\ref{thm:Middle} since the utility of every agent is defined by a single facility which is located at $1/2$. For the same reason, Theorem~\ref{thm:Middle} also implies that the approximation ratio of the mechanism is at most $2$ for the min and max utility classes when $k \geq 2$. 

Now consider any instance $I=(\xbf, \tbf, m, k)$ with $m > k \geq 2$. To show that the mechanism is strategyproof for any utility class, first observe that since it does not take into account the positions of the agents when deciding which subset of facilities to locate and where, the agents have no incentive to misreport their positions. In addition, no agent has any incentive to unilateraly misreport her preferences since any such misreport can only increase the count of facilities she does not approve, and thus her utility cannot increase. 

To show that the mechanism has approximation ratio at most $2$ for the sum utility class, let $S$ be the subset of $k$ facilities chosen by the mechanism, and also let $O$ be the optimal subset of $k$ facilities. We make the following simple observations:
\begin{itemize}

\item For every $j \in S$, it holds that $n_j = \sum_{i \in N} t_{ij}$.

\item Since every $j \in S$ is placed at $1/2$, it holds that $1-d(x_i,y_j) \geq 1/2$ for every agent $i$ that approves $j$. 

\item 
By the definition of the mechanism, we have that $n_j \geq n_o$ for every $j \in S$. 

\item 
Since the maximum utility of any agent is $k$, we have that $W^*(I) \leq k \cdot n_o$.
\end{itemize} 
Putting everything together, we have that 
\begin{align*}
W((k,m)\text{-MM}(I)|I) 
&= \sum_{i \in N} \sum_{j \in S} t_{ij} \cdot \left( 1 - d\left(x_i,\frac12\right) \right) \\
&\geq \frac{1}{2} \sum_{j \in S} \sum_{i \in N} t_{ij} 
= \frac{1}{2} k \cdot n_j 
\geq \frac{1}{2} k \cdot n_o \geq \frac{1}{2} W^*(I),
\end{align*}
and the bound on the approximation ratio follows.
\end{proof}

For completeness, we present a simple instance showing that the $(k,m)$-Middle mechanism is not group-strategyproof when $k\geq 2$ for any of the utility classes we consider. 

\begin{lemma} \label{lem:km-Middle-not-gsp-kgeq2}
For any utility class $C \in \{\text{sum}, \text{min}, \text{max}\}$, the $(k,m)$-Middle mechanism is not group-strategyproof when $k \geq 2$. 
\end{lemma}

\begin{proof}
Consider an instance with $m \geq 4$ facilities, from which we must choose and locate $k \in \{2, \ldots, m-2\}$. There are $n = m$ agents such that every agent approves a different facility; specifically, agent $i \in \{1, \ldots, m\}$ approves only facility $i$. Since $n_i = 1$ for every facility $i$, the $(k,m)$-Middle mechanism can choose any set $S$ of $k$ facilities. Then, for any utility class $C \in \{\text{\em sum}, \text{\em min}, \text{\em max}\}$, every agent approving a facility $j \not\in S$ obtains zero utility. Since $k \leq m-2$, every pair of agents $(i,j)$ approving facilities $i,j \not\in S$ have incentive to form a coalition and change their preferences so that they both approve $i$ and $j$. Such a group misreport would lead to $n_i=n_j=2$, and result in both $i$ and $j$ being part of any set of $k$ facilities chosen by the mechanism, thus showing that the agents have successfully manipulated the mechanism.
\end{proof}

By appropriately extending the proof of Theorem~\ref{thm:known-deterministic-lower}, we can show that, for any $m$ and $k \geq 1$ such that $m \geq 2k$, the approximation ratio of any deterministic mechanism is at least $2$, even when the preferences of the agents are known. As a result, the $(k,m)$-Middle mechanism is the best possible strategyproof deterministic mechanism in terms of approximation in the general and in the known-preferences settings for any such choice of $m$ and $k$. 

\begin{theorem} \label{thm:km-deterministic-LB-2}
For any utility class $C \in \{\text{sum}, \text{min}, \text{max}\}$ and any $m$, $k$ such that $m \geq 2k$, the approximation ratio of every deterministic strategyproof mechanism that locates $k$ out of $m$ facilities is at least $2-\delta$, for any $\delta > 0$, even when the preferences of the agents are known.
\end{theorem}

\begin{proof}
Consider an arbitrary deterministic strategyproof mechanism and the following instance $I$ with $m \geq 2$ facilities, from which we must choose and locate $k \leq m/2$. There are $n=2m$ agents, such that there are exactly two agents that approve only facility $j \in \{1, \ldots, m\}$. For every facility $j$, one of the agents that approve it is positioned at some $\varepsilon \in (0, 1/2)$, while the other such agent is positioned at $1$.
Let $S$ be the subset of $k$ facilities which the mechanism chooses to locate. Clearly, since every agent approves a single facility, her utility is the same under any utility class. In particular, every agent that approves a facility $j \not\in S$ has utility $0$ in $I$, and every pair of agents that approve a facility $j \in S$ have combined utility at most $1+\varepsilon$ (when $j$ is located anywhere in the interval $[\varepsilon, 1]$). Hence, the social welfare of the mechanism in $I$ is at most $(1+\varepsilon)k$.

Now, let us enumerate the facilities not included in $S$ as $\{1, \ldots, m-k\}$. Consider a sequence of instances $I_0=I, I_1, \ldots, I_{m-k}$ such that instance $I_j$, $j \in \{1, \ldots m-k\}$ is obtained from instance $I_{j-1}$ by moving the agent that is positioned at $\varepsilon$ and approves facility $j$ to $1$. Since $I_0=I$ and every pair of instances $(I_j, I_{j-1})$ differ only on the position of a single agent that approves a facility not in $S$, the mechanism must choose to locate the same set $S$ of $k$ facilities so that the utility of the agents that are moved is not increased from zero to positive; otherwise, the mechanism would not be strategyproof. 

In the last instance $I_{m-k}$ of this sequence, all the agents that approve facilities not in $S$ are located at $1$. Since $m \geq 2k$, there exists a subset $S'$ of $k$ facilities such that $S \cap S' = \varnothing$ which can be located at $1$ to achieve a social welfare of $2k$; each of the two agents whose facility is chosen has utility equal to $1$. Since the agents approving facilities in $S$ are at the same positions as in $I$, by locating the set $S$ of facilities, the social welfare of the mechanism in $I_{m-k}$ is again at most $(1+\varepsilon)k$. The bound of $2-\delta$ on the approximation ratio of the mechanism follows by selecting $\varepsilon$ to be arbitrarily small. 
\end{proof}

\section{Conclusion and open problems}
There are several interesting problems that either remain open or arise from our work. The first natural direction is to tighten our results for deterministic and randomized mechanisms for the different settings we have considered. For deterministic mechanisms, while the general and the known-preferences settings are resolved by our work, it would still be quite interesting to close the gap between $13/11$ and $2$ for the known-positions setting. For randomized mechanisms, the most intriguing open question is whether there exists such a mechanism with approximation ratio significantly smaller than $2$ in the general setting. An obvious candidate is the RD mechanism that we presented in the context of the known-positions setting. Unfortunately, the particular variant of RD is no longer strategyproof when both the positions and the preferences of the agents are private, as shown by the following lemma.

\begin{figure}[t]
\centering
\includegraphics[scale=0.74]{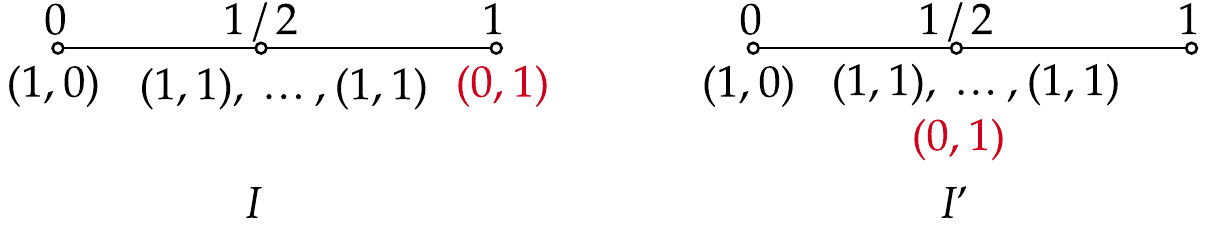}
\caption{The two instances used in the proof of Lemma~\ref{lem:RD-not-sp} to show that RD is not strategyproof in the general setting. In instance $I$, RD chooses facility $1$ whenever an agent with preferences $(1,1)$ is chosen as the dictator. This gives the agent with preferences $(1,0)$ marked in red incentive to misreport her position, thus leading to instance $I'$, where the tie is broken in favor of facility 2.}
\label{fig:RD}
\end{figure}

\begin{lemma} \label{lem:RD-not-sp}
RD is not strategyproof in the general setting. 
\end{lemma}

\begin{proof}
Let $n \geq 4$, and consider the following instance $I$ depicted in the left part of Figure~\ref{fig:RD}. There is an agent with preferences $(1,0)$ positioned at $0$, $n-2$ agents with preferences $(1,1)$ positioned at $1/2$, and one agent with preferences $(0,1)$ positioned at $1$. Since the maximum possible social welfare by placing facility $1$ is the same as the maximum possible social welfare by placing facility $2$, we can without loss of generality assume that the optimal facility is $1$, and thus when an agent with preferences $(1,1)$ is randomly chosen, facility $1$ is placed at her position; to avoid ties altogether, we could move the $n-2$ agents that are now positioned at $1/2$ to $1/2-\varepsilon$, for some arbitrarily small $\varepsilon > 0$. Observe that the expected utility of the agent $i$ that has preferences $(0,1)$ positioned at $1$ is $\frac{1}{n}$, as she only gets some utility when she is chosen.

Now consider the instance $I'$ that is obtained from $I$ by moving agent $i$ from $1$ to $1/2$; this instance is depicted in the right part of Figure~\ref{fig:RD}. Since the optimal facility in $I'$ is $2$, when an agent with preferences $(1,1)$ is randomly chosen as the dictator, facility $2$ is placed at her position. Therefore, agent $i$ has incentive to deviate to $1/2$ and change instance $I$ to $I'$ in order for her expected utility to become $(1-\frac{1}{n})\cdot \frac12 = \frac{n-1}{2n} > \frac{1}{n}$, thus proving that RD is not strategyproof when the agents can misreport their positions. 
\end{proof}

Intuitively, the reason that makes RD manipulable in the general setting is the tie-breaking rule that we use for the agents who approve both facilities; recall that such ties are broken in favor of the optimal facility. Breaking ties in this way is crucial for our characterization of the worst-case instances in Section~\ref{sec:positions}, but can be exploited by agents who are allowed to misreport their positions. Designing a variant of RD that is strategyproof in the general setting is straightforward, for example, by breaking ties between the two facilities equiprobably. Importantly however, the aforementioned characterization no longer holds in that case, which makes the analysis much more challenging. As a matter of fact, we can show that for any variant that uses a fixed probabilistic tie-breaking rule, the approximation ratio is \emph{strictly larger} than $3/2$! In the following theorem, we show this for the version of RD that breaks ties by locating facility $1$ with probability $p$ and facility $2$ with probability $1-p$; we refer to this mechanism as $p$-RD.

\begin{lemma}
\label{lem:pRD}
$p$-RD has approximation ratio at least $1.518$, for every fixed $p \in [0,1]$.
\end{lemma}

\begin{proof}
Without loss of generality assume that $p \in [0,1/2]$, and consider the following instance. There are $30$ agents positioned at 0: $15$ of them have preferences $(1,1)$ and the remaining $15$ have preferences $(1,0)$. Furthermore, there are $20$ agents positioned at 1: $10$ of them have preferences $(1,0)$ and the remaining $10$ have preferences $(0,1)$. Observe that the optimal solution locates facility 1 at 0 and achieves social welfare $30$. The expected social welfare of $p$-RD is $\frac{(3+p)\cdot 15^2 + 2\cdot 10^2}{50}$. Hence, the approximation ratio of $p$-RD is $\frac{1500}{(3+p)\cdot 15^2 + 200}$, which is larger than $\frac{120}{79} > 1.518$ for every $p \in [0, 1/2]$.
\end{proof}

Finding the exact approximation ratio of $p$-RD is an intriguing open question. Perhaps more interestingly, one can define yet another strategyproof variant of RD, whose approximation ratio is not ruled out by Lemma~\ref{lem:pRD}. For example, we can count how many agents approve each facility and break ties {\em proportionally} to those numbers. It is quite easy to observe that the RD mechanism using the proportional tie-breaking rule is strategyproof for the general setting, and is a promising candidate for achieving a better approximation ratio. To this end, we state the following conjecture.
\begin{conjecture}
\label{ref:conj-RD}
For the general setting, the RD mechanism with the proportional tie-breaking rule has approximation ratio $3/2$.
\end{conjecture}

Besides strengthening our results, there are several meaningful extensions of our model that could be the subject of future work. For the $k$ out of $m$ facilities setting, while we have made an important first step, there is still significant work to be done, particularly in the known-positions setting. One could also consider several different variants of our basis model. For example, the agents may have fractional preferences rather than approval preferences (that is, each agent may assign weights in $[0,1]$ to the facilities, instead of weights in $\{0,1\}$). Other possible variants include settings in which some facilities are \emph{obnoxious} \citep{YMZ,FJ15}, meaning that agents would like to be far from them if they are built, and discrete settings in which the facilities can only be built at predefined locations on the line (e.g., see \citep{DFMN12,FFG16,SV15,SV16}). 
Finally, an interesting generalization of our problem is when every facility comes at a different cost, and the objective is to maximize the social welfare by choosing and locating $k$ facilities under the constraint that their accumulated costs is below a predefined budget. This latter setting is directly motivated by \emph{participatory budgeting}, which has recently drawn the attention of the computational social choice community~\citep{benade2020preference,aziz2020participatory}.

\bibliographystyle{plainnat}
\bibliography{references}

\end{document}